\documentclass[letter,11pt]{llncs}
\usepackage[margin=1in]{geometry}

\usepackage{amssymb}
\usepackage{graphicx}

\usepackage{url}
\urldef{\mailsa}\path|xlma@pku.edu.cn|   

\urldef{\mailsb}\path|{gzhou6, jwang112, jianpeng, hanj}@illinois.edu|    
\newcommand{\keywords}[1]{\par\addvspace\baselineskip
\noindent\keywordname\enspace\ignorespaces#1}

\let\llncssubparagraph\subparagraph
\let\subparagraph\paragraph
\let\subparagraph\llncssubparagraph
\usepackage[colorlinks,
           linkcolor=red,
           anchorcolor=blue,
           citecolor=blue           
           ]{hyperref}

\usepackage[]{algorithm2e}
\usepackage{subfig}
\usepackage[misc]{ifsym}
\RestyleAlgo{boxruled}

\usepackage{authblk}

\begin{document}

\mainmatter  

\setlength{\intextsep}{0pt}

\setlength{\floatsep}{5pt plus 1.0pt minus 10.0pt}
\title{Complexes Detection in Biological Networks via Diversified Dense Subgraphs Mining}

\titlerunning{Complexes Detection by Discovering Diversified Dense Subgraphs}

%
%
\author{Xiuli Ma\inst{1}\textsuperscript{(\Letter)}  \and Guangyu Zhou\inst{2} \and Jingjing Wang\inst{2} \and Jian Peng\inst{2}\textsuperscript{(\Letter)} \and Jiawei Han\inst{2}\textsuperscript{(\Letter)} }

\renewcommand\Authands{ and }

%
\authorrunning{Complexes Detection by Discovering Diversified Dense Subgraphs}

\institute{School of EECS, Key Laboratory of Machine Perception (MOE), Peking University, Beijing, China\\ \mailsa\\ \and
Dept. of Computer Science, University of Illinois at Urbana-Champaign, Urbana, IL, US\\\mailsb
}

%
%

\toctitle{Lecture Notes in Computer Science}
\tocauthor{Authors' Instructions}
\maketitle

\begin{abstract}
Protein-protein interaction (PPI) networks, providing a comprehensive landscape of protein interacting patterns, enable us to explore biological processes and cellular components at multiple resolutions. For a biological process, a number of proteins need to work together to perform the job. Proteins densely interact with each other, forming large molecular machines or cellular building blocks. Identification of such densely interconnected clusters or protein complexes from PPI networks enables us to obtain a better understanding of the hierarchy and organization of biological processes and cellular components. Most existing methods apply efficient graph clustering algorithms on PPI networks, often failing to detect possible densely connected subgraphs and overlapped subgraphs. Besides clustering-based methods, dense subgraph enumeration methods have also been used, which aim to find all densely connected protein sets. However, such methods are not practically tractable even on a small yeast PPI network, due to high computational complexity. In this paper, we introduce a novel approximate algorithm to efficiently enumerate putative protein complexes from biological networks. The key insight of our algorithm is that we do not need to enumerate all dense subgraphs. Instead we only need to find a small subset of subgraphs that cover as many proteins as possible, meanwhile have minimal overlap among themselves. The problem is formulated as finding a diverse set of dense subgraphs, where we develop highly effective pruning techniques to guarantee efficiency. To handle large networks, we take a divide-and-conquer approach to speed up the algorithm in a distributed manner. By comparing with existing clustering and dense subgraph-based algorithms on several human and yeast PPI networks, we demonstrate that our method can detect more putative protein complexes and achieves better prediction accuracy. \footnote{This paper was selected for oral presentation at RECOMB 2016 and an abstract is published in the conference proceedings}
\end{abstract}

\section{Introduction}

Recent developments in high-throughput experimental procedures have resulted in very large repositories of protein-protein interaction (PPI) and genetic interaction networks\cite{gavin2006proteome,krogan2006global,uetz2000comprehensive,ito2001comprehensive}. These interaction networks provide us a molecular landscape that defines fundamental biological processes in living cells. Many proteins need to interact closely with each other and form larger molecular machines to perform complex molecular functions. One major type of such molecular machines is protein complexes, in which proteins are densely packed  together with very strong or even permanent interactions. As most cellular tasks are performed not by individual proteins or genes, but by groups of functionally associated proteins or genes \cite{adamcsek2006cfinder}, identifying such tightly knit groups is crucial to understand and explore the structural and functional properties of the biological networks. Thus detection of protein complexes \cite{bader2003automated,van2001graph,wu2009core,nepusz2012detecting} from PPI networks is a fundamental problem in network biology, which has been normally formulated as a problem to find subsets of densely interconnected proteins or clusters in networks. 

Although quite a few methods on complex detection claimed that, protein complexes correspond to ``densely connected regions'' or ``dense subgraphs'' in PPI networks \cite{adamcsek2006cfinder,bader2003automated,king2004protein,liu2009complex,nepusz2012detecting}, most of them generally model the problem as graph clustering. They typically first partition the networks into clusters and then post-process or filter the clusters with density thresholds. Considering that previous research have identified the existence of overlapped complexes, some recent algorithms have been proposed to detect overlapping clusters \cite{ahn2010link,palla2005uncovering,nepusz2012detecting,adamcsek2006cfinder}. However, no matter with or without overlapping considered, the objective function of such clustering methods is to maximize the difference between intra- and inter-connectedness of clusters. Such relative models have no preset level for what is sufficiently dense, which may lead to some results not dense enough, or lose some really dense or meaningful results.

There are also algorithms that predict protein complexes by identifying dense subgraphs in PPI networks \cite{adamcsek2006cfinder,bader2003automated,pizzuti2014algorithms}. Such algorithms can enumerate all the dense subgraphs however being computational intractable. Moreover, there is no need to find all dense subgraphs as there will be substantial redundancy. The number of all the dense subgraphs could also be tremendous such that the biologists get overwhelmed. Many small subgraphs may be completely covered in their supersets which are also identified as dense subgraphs. Intuitively dense subgraphs with maximal cardinality, which cannot be further extended, should be considered. Furthermore, although overlap is allowed between subgraphs, very large overlap means redundancy. A diverse set of dense subgraphs which cover as many proteins as possible in the networks is more desirable. Overall, meaningful subgraphs are expected to satisfy the following requirements: they cannot be too restrictive as cliques; they should be based on density of links; they should not contain any cut nodes or cut edges; and overlap is allowed but redundancy should be minimized.
Recently, some mining algorithms have been proposed or can be adopted to detect a set of dense subgraphs \cite{balalau2015finding,tsourakakis2013denser}, but in most cases they can only identify non-overlapping dense subgraphs. They typically discover a single densest subgraph, then remove its nodes and edges from the network, and iterate until enough subgraphs are found or no edges are left. For the density, some methods propose to maximize the average degree \cite{balalau2015finding} in subgraphs, which tends to find large subgraphs with only modest density. Some other methods \cite{adamcsek2006cfinder,bader2003automated} first identify ``cliques'' or ``cores'' and then merge these cliques. Such heuristics are too simple and may miss important dense subgraphs. 

In this paper, we propose to detect protein complexes by finding diversified dense subgraphs, with an explicit definition for density, diversity and a set of diversified results, and a unified method during enumeration. As far as we know, no existing work takes such diversifying into the problem of complexes detection. Specially, the key component of our algorithm is a set of efficient search trees that compactly traverse all dense subgraphs by a depth-first construction. We then define a node-specific potential to guide the search process and develop efficient pruning techniques based on both density and diversity of subgraphs.
In this way, we extract the diverse dense subgraphs ``on-the-fly'' during the enumeration of the maximal dense subgraphs, thus greatly improve the scalability of the algorithm. Finally, we speed up the algorithm in parallel to handle large-scale networks. 
We extensively evaluate the effectiveness and efficiency of our method on several PPI networks from yeast and human. In all networks, our approach detects more putative dense complexes, and achieves higher accuracy and better one-to-one mapping with reference complexes than several state-of-the-art algorithms.

\section{Problem Definition}

Let $G=(V, E)$ be an undirected graph with a vertex set $V$ and an edge set $E\subseteq V\times V$. $w(e)$ is the weight of edge $e\in E$. Weights are normalized to the range $[0, 1]$. We treat unweighted graphs as the special case where all weights are equal to 1. For a set of vertices $S \subseteq V$, we denote the subgraph of $G$ induced by $S$ as $G(S) = (S, E(S))$, where $E(S)=\{(u,v)\in E | u, v\in S\}.$

\begin{definition} \textbf{Density:} The density of subgraph $S$, $den(S)$, is the ratio of the total weight of edges in $E(S)$ to the number of possible edges among $|S|$ vertices. The density of $S$ is defined as {\small $den(S) = \sum_{u,v\in S}w(u,v) / {|S| \choose 2}$}.

\end{definition}
Note that if the graph is unweighted, the numerator is simply the number of actual edges. No matter weighted or not, the maximum possible density is $1$.

\begin{definition} \textbf{Dense subgraphs:}
Given a graph $G$ and a density threshold $\theta$, $0 \leq \theta \leq 1$, an induced subgraph S of $G$ is called a dense subgraph if it is connected and its density is no less than $\theta$. A dense subgraph S is called a maximal dense subgraph if there exists no dense subgraph $S'$ in $G$ such that $S\subseteq S'$ .
\end{definition}

In this paper, we diversify the maximal dense subgraphs both for efficiency and conciseness. Specially, we diversify the results to cover the most nodes.

\begin{definition}\textbf{Coverage:} 
 Given a set of dense subgraphs $D = \{S_1, S_2, ...\}$ in graph G, the coverage of D, denoted by $cov(D)$, is the number of nodes in G covered by the dense subgraphs in D, i.e., 
 {\small $cov(D) = |\cup_{S_i\in D}S_i|$}.
 \end{definition}

\subsubsection* {Problem Statement.} 
Given a graph $G$, a density threshold  $\theta$ and an integer $k$, the problem of detecting $k$ diversified maximal dense subgraphs is to discover a set $D$, such that each $S_i\in D$ is a maximal dense subgraph with density no less than $\theta$ in $G$, $|D| \leq k$, and $cov(D)$ is maximized. $D$ is called the set of diversified maximal dense subgraphs.

Figure 1 illustrates three complexes detected by our algorithm from the Krogan core dataset with density threshold 0.7. Shaded areas represent the dense subgraphs. Red, blue and green nodes represent diverse nodes of the complexes, respectively; yellow nodes represent overlapped nodes. The set of two diversified maximal dense subgraphs is $\{A, C\}$.

Considering that it may be hard for the users to specify a suitable value for $k$, we have an alternative problem definition.

\begin{definition} \textbf{Diversity:} The diversity of a set of dense subgraphs D, div(D), is the ratio of the coverage of D to the number of dense subgraphs in D:{ \small $div(D) = cov(D)/|D|$}.

\end{definition}
\subsubsection* {cov/k alternative problem statement.}
Given a graph $G$, a density threshold  $\theta$ and a diversity threshold $\gamma$, the problem ``$cov/k$'' is to discover a set $D$ of maximal dense subgraphs, such that each $S\in D$ is a maximal dense subgraph with density no less than $\theta$ in $G$, and $|D|$ is maximum that $div(D)\geq \gamma$. In essence, it means that, each subgraph in the result set $D$ need to cover at least $\gamma$ nodes in average. 

\begin{figure}
\centering
\includegraphics[height=7cm]{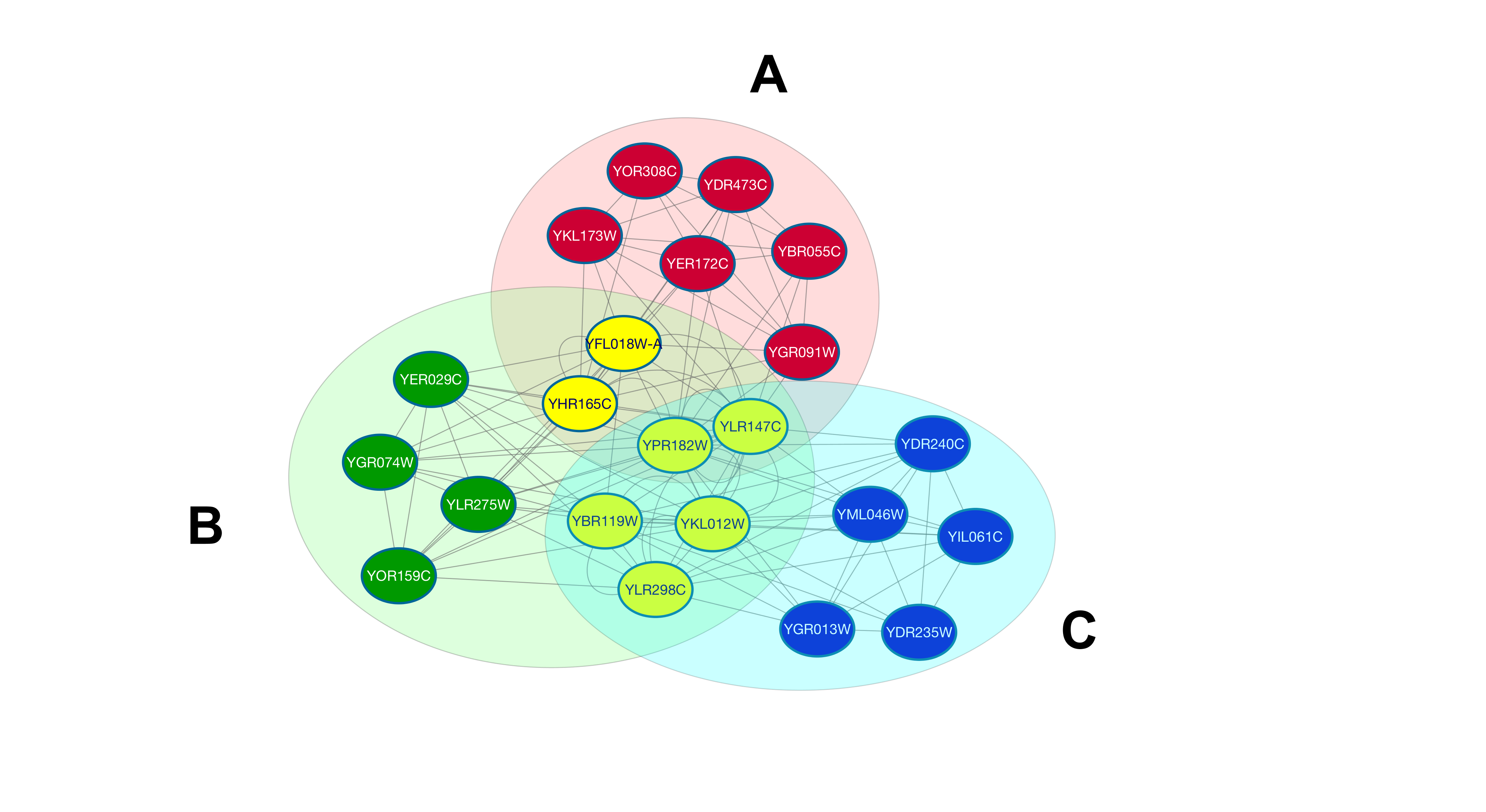}
\caption{Illustration of three complexes detected by our algorithm from the Krogan core dataset with density threshold 0.7.\\}

\label{fig:overlap}
\end{figure}

So our focus is on maximizing the coverage subject to a cardinality constraint. Unfortunately, this problem is NP-hard\cite{feige1998threshold}. Note that a naive solution to the problem is to firstly enumerate all the dense subgraphs, pick out the maximal ones, and then find $k$ of them that cover most nodes in the graph by using the approximate greedy max $k$-cover algorithm\cite{nemhauser1978analysis}. However, such a solution is impractical when the graph is large because the number of dense subgraphs is exponential in the graph size. Also, repeated scans of inputs for max $k$-cover is inefficient. In this paper, we integrate the searching and diversifying tightly into one unified process.

\section{Diversifying Maximal Dense Subgraphs}
As the classic recursive backtracking procedure\cite{tomita2006worst,eppstein2010listing} that enumerates all the dense subgraphs in a graph is impractical, we target at the dense subgraphs with both maximal cardinality and distinctive contribution on coverage. Before diving into the algorithm, let us first identify two properties, one for density, the other for coverage.

\subsection{Properties of Density and Coverage}
\subsubsection{The Pseudo Anti-Monotonicity Property}
Note that in our definition of dense subgraphs, the family of dense subgraphs no longer satisfies the anti-monotone property. However, we can still get a rather general while effective-for-pruning property. We call it pseudo anti-monotonicity because it is not as rigid as anti-monotonicity however at least holds in one case.

\begin{theorem} [Pseudo anti-monotonicity]
Let $G$ be a graph. Given a density threshold $\theta$, if $G'$ is a dense subgraph, there is at least one subset of $G'$ whose density is no less than $\theta$.
\end{theorem}

\begin{proof}
If we can find such one subset of $G'$, we can surely prove it. Let $v$ be a vertex in $G'$ with the degree no greater than the average of the degrees of vertices in $G'$, the density of $G'\setminus \{v\}$ is no less than the density of $G'$ \cite{uno2010efficient}. However, it is possibly such a subset is not connected at all. In this case, $G'$ has not any subset that is both dense and connected. In this paper, we would rather give up such dense subgraphs, as they are meaningless in biological networks. We do not expect those complexes with cut-nodes (whose removal would disconnect a complex).
\end{proof}

Note that this theorem introduces an adjacency relationship on dense subgraphs. Since any dense subgraph $G'$ has such a vertex $v$, we can remove the vertices of $G'$ iteratively until $G'$ is empty, passing through only dense subgraphs. That is also to say, each dense subgraph can be incrementally constructed, one vertex at a time. However, we need an appropriate ordering of the vertices to ensure the monotone decreasing of density along the growth in order to prune with density.

\subsubsection{The Submodularity Property}

Actually, by diversifying, our goal is to select a manageable subset of maximal dense subgraphs that are most representative according to the objective function coverage. Many natural notions of ``representativeness'' satisfy submodularity, an intuitive notion of diminishing returns, also is coverage \cite{badanidiyuru2014streaming}.

Firstly, coverage is naturally associated with a marginal gain, 
\begin{equation}
\small
\Delta cov(S\mid D)=cov(D \cup \{S\}) - cov(D)
\end{equation}
where $D$ is a set of maximal dense subgraphs, $S$ is a maximal dense subgraph, which qualifies the increase in coverage obtained when adding $S$ to set $D$. 

\begin{theorem}
The coverage function is monotone and submodular. Coverage is monotone as for all S and D it holds that $\Delta cov(S\mid D) \geq 0$. 
Coverage is submodular as for all $A \subseteq B \subseteq D$ and $S\in D\setminus B$ the following diminishing returns condition holds:
\begin{equation}
\small
\Delta cov(S\mid A) \geq \Delta cov(S\mid B)
\end{equation}
\end{theorem}

As observed in \cite{badanidiyuru2014streaming}, the key reason why the classical greedy algorithm for submodular maximization works is that at every iteration, an element is identified that reduces the ``gap'' to the optimal solution by a significant amount. More formally, in our scenario, if $D_i$ is the set of first $i$ maximal dense subgraphs picked by the greedy algorithm, then the marginal gain of the next maximal dense subgraph {\small$S_{i+1}$} is at least {\small $(OPT - cov(D_i))/(k - |D_i|)$}, where $OPT$ is the coverage for the optimal solution, $k$ is the cardinality constraint of the problem. In this way, we can ``sieve'' out maximal dense subgraphs with large marginal values. 

However, this will require that we know (a good optimization to) the value of the optimal solution $OPT$, which is hard to obtain before getting the solution. Actually, in order to get a very crude estimate on $OPT$, it is enough to know the maximum marginal gain of any single result, {\small $m = max_{S\in D} cov(\{S\})$}. Then, from submodularity, we have that {\small $m \leq OPT \leq km$}. Once we get this crude upper bound $km$ on $OPT$, we can refine it. 

Note that, this assumes that the value $m$ is known at the very beginning of the algorithm. That is to say, if we first find the maximum dense subgraph, we can pick out the dense subgraphs with enough marginal gain on the fly, thus diversify the results during the enumeration. Here, a maximum dense subgraph of a graph is a dense subgraph having maximum size, while a maximal dense subgraph is a dense subgraph that is not a subset of any other dense subgraph.
The crucial point is, we must find the maximum dense subgraph first.

\subsection{Finding the Maximum Dense Subgraph First}

As analyzed above, knowing the maximum dense subgraph first will help us decide online whether any resulting subgraph has sufficient marginal gain. So, we need an ordering to build the search tree, finding the maximum dense subgraph first, meanwhile, getting all the other maximal dense subgraphs efficiently.

With this goal on our mind, we adopt the greedy randomized adaptive search procedure (GRASP)\cite{abello2002massive}. As finding the maximum dense subgraphs is computationally intractable, it has been proposed to find the maximal dense subgraph. When a maximal dense subgraph is iteratively constructed, the choice of next vertex to be added can be determined by ordering all candidates by potential \cite{abello2002massive}. Specially, the maximum dense subgraph corresponds to the leftmost branch in the search tree.

In order to introduce the notion of potential, we first give several preliminaries \cite{abello2002massive}.
A vertex $x$ is called a $\gamma$-vertex with respect to a dense subgraph $S$, if $G(S\cup \{x\})$ is a dense subgraph given the density threshold $\gamma$. The set of $\gamma$-vertices with respect to $S$ is denoted by $N_{\gamma}(S)$.

The potential of a subgraph R is
\begin{equation}
\small
\phi(R) = |E(R)| -\gamma{|R| \choose 2}
\end{equation}

The potential of a set $R$ with respect to a disjoint set $S$ is
\begin{equation}
\small
\phi_S(R) = \phi(S\cup R)
\end{equation}

Assume $S$ is a dense subgraph. We seek a vertex $x\in N_{\gamma}(S)$ to be added to S. One strategy for selecting $x$ is to measure the effect of its selection on the potential of the other vertices in $N_{\gamma}(S)$. To accomplish this, define the potential difference of a vertex $y\in N_{\gamma}(S)\setminus\{x\}$ by selecting $x$ to be
\begin{equation}
\small
\delta_{S,x}(y) = \phi_{S\cup \{x\}}(\{y\}) - \phi_S(\{y\})
\end{equation}

The total effect on the potentials, caused by the selection of $x$, on the remaining vertices of $N_{\gamma}(S)$ is

\begin{equation}
\small
\Delta_{S,x} = \sum_{y\in N_{\gamma}\setminus\{x\}}
\delta_{S,x}(y)
= |N_{\gamma}(\{x\})| + |N_{\gamma}(S)|(deg(x)|_S - \gamma(|S| + 1))
\end{equation}

We call this metric the potential of $x$ with respect to $S$. The vertex $x$ that maximizes this metric is the one with a high number of $\gamma$-neighbors and with high degree with respect to $S$. A greedy algorithm that recursively selects such a vertex will eventually terminate with a maximum dense subgraph.

\subsection{Detecting Diversified Maximal Dense Subgraphs}
Once we compute the potential for each candidate of current subgraph, we can expand the candidates by the potential order. We will show this order can facilitate the pruning by both density and diversity. 

As we will outline in the algorithm, the process of pruning and diversifying are embedded in the backtrack-based algorithm. Specially, pruning by both density and diversity is performed when the process is starting a new search subtree, while diversifying is executed when a maximal dense subgraph is found (on a non-pruned branch of the search tree).

\subsubsection{Pruning by Density}
\begin{theorem}
Let $S$ denote a dense subgraph. $k$, $i$ are in $S$'s {\small $N_{\gamma}(S)$}.  
{\small $\Delta_{S,k} \geq \Delta_{S,i}$}. 
Then, {\small $den(S \cup {k}\cup {i})\leq den(S \cup {k})$}.
\end{theorem}
\begin{proof}
In Appendix.
\end{proof}

This theorem is to say, after we sort the candidates by the order of potential, a vertex can only appear in the subtree of another vertex with higher ordering, not vice versa. If the enumeration tree grows in this way, the density will decrease along the path in depth-first search, we can safely prune an expansion if the density is less than the density threshold. 

\subsubsection{Pruning by Diversity}
Next, we discuss how to employ diversity to speedup the search and how to diversify.

\paragraph{Diversifying.} 
If having an estimate value {\small$v$} for {\small $OPT$}, as according to submodularity, we can get a crude upper bound {\small $km$} on {\small$OPT$}, so can immediately refine it. Consider the following set,
\begin{equation}
\small
O = \{ (1+\epsilon)^i \mid i\in Z, m \leq (1+\epsilon)^i \leq km \}
\end{equation}

At least one of the thresholds {\small$v\in O$} should be a pretty good estimate of {\small $OPT$}, i.e., there should exist at least some {\small$v\in O$} such that {\small$(1 - \epsilon)OPT\leq v \leq OPT$}. That means, we could run the algorithm once for each value {\small $v\in O$}, requiring multiple passes over the set of dense subgraphs. In fact, instead of using multiple passes, a single pass is enough: We simply run several copies of the algorithm in parallel, producing one candidate solution for each threshold {\small$v\in O$}. As final output, we return the best solution obtained. This procedure is called SIEVE. For space limit, we do not illustrate SIEVE\cite{badanidiyuru2014streaming}.

This algorithm assumes that we get the maximum dense subgraph at the beginning. However, even if we cannot get the maximum one at start, we can maintain $m$ that holds the length of current maximum dense subgraph. Whenever $m$ gets updated, the algorithm lazily instantiate the set $O_i$ and delete all thresholds outside $O_i$. \\

\paragraph{Pruning.} 
Different from diversifying at leaf nodes, for pruning, we have to decide immediately whether to prune current expansion or not, based on its marginal value. On one hand, we will get its upper bound. On the other hand, we want to have the threshold lower-bounded.

Assume we have found a dense subgraph $S$, and want to decide whether we should expand it further. Let $d$ be the depth of its subtree. Remember that we have got the maximum dense subgraph at first, so we can get the upper bound of $d$, $\bar{d}$, which is the size of the maximum dense subgraph. Then, the size of each dense subgraph to be generated from $S$ is at most  $\bar{d}$. Assume $S$ has already got $l$ diverse vertices. As $S$ has at most $(\bar{d}-|S|)$ nodes to expand, its marginal gain is at most $l + (\bar{d} -|S|)$.

On the other hand, we want to have its threshold lower-bounded. As analyzed at the beginning of this section, the marginal gain of each result should be at least $(v/2-cov(D))/(k-|D|)$, where $D$ is the set of maximal dense subgraphs we have got, $v$ is an estimate for $OPT$ which is the coverage for the optimal solution, $k$ is the cardinality constraint of the problem. The point is, as we have to decide immediately whether to prune current expansion or not, cannot try different or multiple values of $v$, we give our suggestion on selection of $v$ based on extensive experiments. Our suggestion is to select between 10\% and 30\% of the last $v$, which is $km$, to be $v$. With this value, we can determine the threshold for marginal gain. We also find, the last $v$ leads to the highest threshold thus the strictest pruning, which is equal to retain only the leftmost branch while pruning all the other branches in each search tree.  For higher efficiency, we use the last $v$ in most evaluations.

In summary, pruning is for the internal nodes, diversifying is for the leaves. In an internal vertex, we will prune a subtree if its predicted marginal gain is lower than the threshold based on a fixed $v$. If a branch arrives the leaf, we will put it into every $S_v$ for $v\in O_i$ if $S$ has the marginal value no lower than the threshold of $S_v$. Finally, it outputs the best solution among $S_v$.

\subsection{Outline of the Whole Algorithm}
Our algorithm grows dense subgraphs from seeds. Initially, it sorts the nodes by the global degrees, and selects the node with the highest degree, or the largest sum of weighted connections in weighted graph, as the first seed. The algorithm grows the dense subgraphs from it using a backtracking-based growth procedure. Whenever the growth process finishes, the algorithm selects the next seed by considering all the nodes that have not been covered in any of the dense subgraphs found so far and taking the one with the highest order among the remaining seeds again. The entire procedure terminates when there are no nodes remaining to consider.
\begin{figure}
\centering
\subfloat[]{\includegraphics[width=0.26\textwidth]{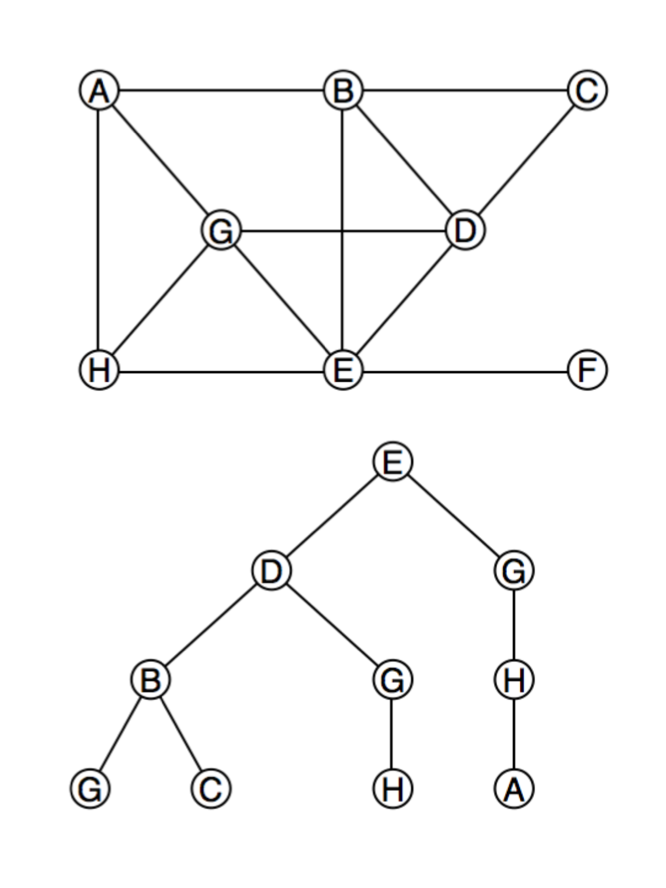}\label{fig:enum1}}
\hspace{2cm}
\subfloat[]{\includegraphics[width=0.5\textwidth]{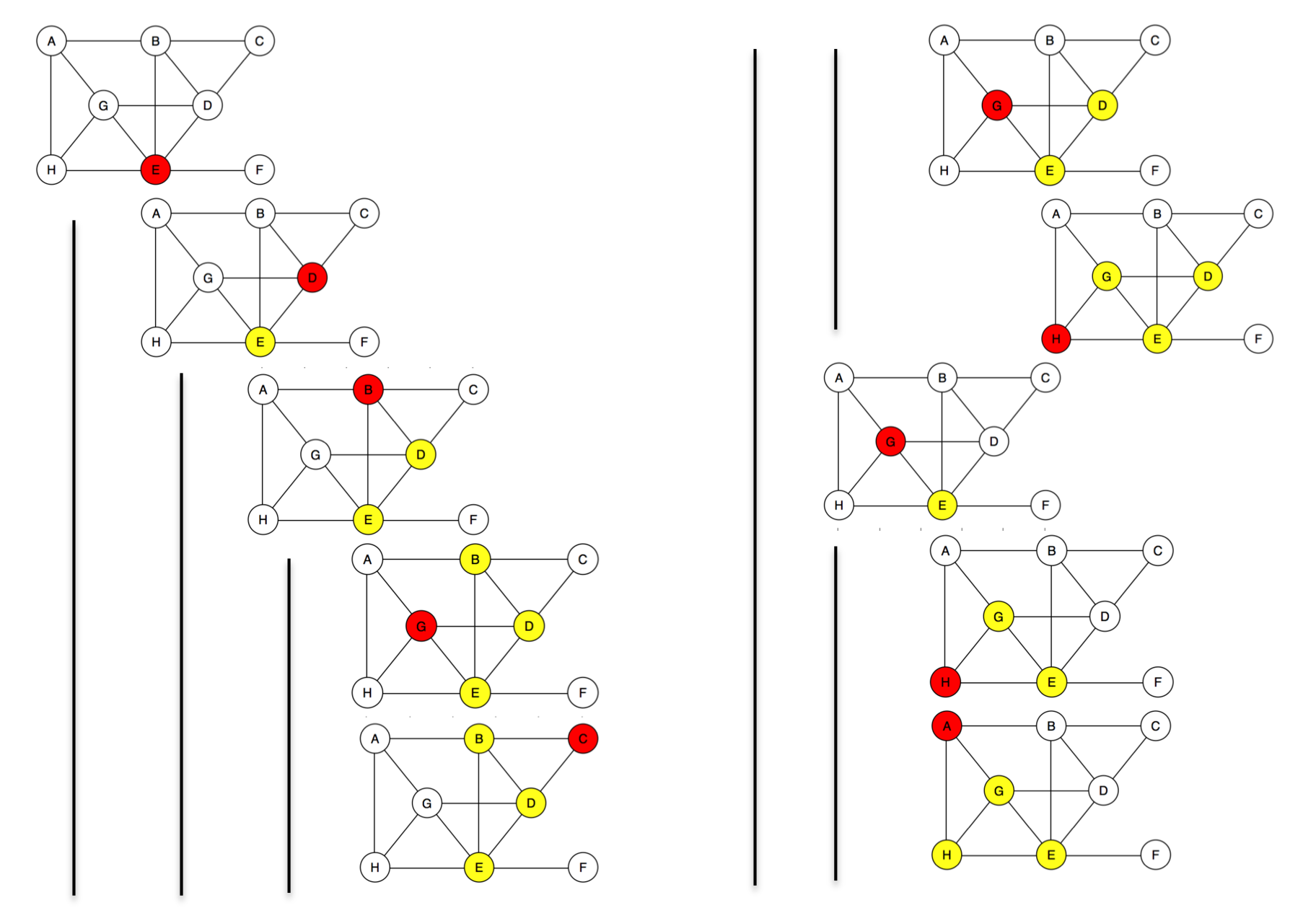}\label{fig:enum23}}
\caption{Illustration the maximal dense subgraph detection process.}
\label{fig:enum}
\end{figure}

A step-by-step description of the growth process starting from $Q$ is follows. Step1: Find all the candidates. Step 2: Order all the candidates according to potential. Step 3: For each candidate $r$ (in the decreasing order), $Q'=Q \cup \{r\}$. If $den(Q')\geq \theta$, and predicted marginal gain is no less than the threshold of diversity, let $Q_{t+1} = Q'$, and iteratively grow from $Q_{t+1}$. Otherwise, declare $Q'$ a maximal dense subgraph. Then it includes $Q'$ into every $S_v$ for $v\in O_i$ if $Q'$ has the marginal value $(v/2-|cov(S_v)|)/(k-|S_v|)$ to $S_v$. Finally, it outputs the best solution among $S_v$. The pseudo code of the algorithm is as follows.

Note that for each result to be generated uniquely, a vertex can appear in the subtree of another vertex having higher ordering than it, but not vice versa. So, for a subgraph $Q$, by $Neighbor(Q)$, we only consider those candidates that are neighbors of $Q$ excluding those having a higher ordering than its nodes. 

The main procedure MDS is as follows.\\
\IncMargin{1em}
\begin{algorithm}
\SetKwData{Left}{left}\SetKwData{This}{this}\SetKwData{Up}{up}
\SetKwFunction{Union}{Union}\SetKwFunction{SortByDeg}{SortByDeg}

\SetKwInOut{Input}{input}\SetKwInOut{Output}{output}
\Input{$G(V,E), \theta$ (density threshold), $m, k$}
\Output{List of maximal dense subgraphs}
\BlankLine
\SortByDeg{v} \\
\ForEach{$v_i \in V$}{
  $S = enum(\{v_i\}, Neighbor(\{v_i\}), G(V,E), \theta, m, k$)

}
return $argmax_{s_i \in S} cov(s_i) $
\caption{{\sc Diversified MDS}}
\end{algorithm}\DecMargin{1em}

The sub-procedure Enumeration is as follows.\\
\IncMargin{1em}
\begin{algorithm}[H]
\SetKwData{Left}{left}\SetKwData{This}{this}\SetKwData{Up}{up}
\SetKwFunction{Union}{Union}\SetKwFunction{graspSort}{graspSort}

\SetKwInOut{Input}{input}\SetKwInOut{Output}{output}
\Input{Q (current subgraph), R (remaining candidate sets)}
\Output{S (Sieve)}
\BlankLine

$enum(Q, R, G(V,E), \theta, m, k)$ \\
$isLeaf = True$

\If{$R.size == 0$}{
  $Sieve(Q, m, k)$
}
\BlankLine

\graspSort{$Q, R$} \\
\BlankLine

\ForEach {$r \in R$}{
      $Q' = Q + \{r\}$\\
      $R' = Neighbor(Q') - Q'$\\
      \If {$density(Q') > \theta$ and not $shouldPruning(Q', m, k, v^* = 0)$}
      {
        $isLeaf = False$\\
        $enum(Q', R', G(V,E), \theta, m, k)$
      }
}
\If {isLeaf}{
      $Sieve(Q, m, k)$
}

\caption{{\sc Enumeration}}
\end{algorithm}\DecMargin{1em}

Actually, diversification can be replaced by greedy max cover algorithm thus can be separated from the MAIN algorithm. Therefore we use MDS to denote the algorithm of MDS detection without diversification.

Figure~\ref{fig:enum}  illustrates the maximal dense subgraph detection process. From the seed $E$, the algorithm greedily generate the maximum dense subgraph EDBG. After that, the maximal dense subgraphs EDBC, EDGH and EGHA are generated. All the other branches are pruned because they do not bring diverse contribution. For space limit, diversifying is not illustrated.

\subsection{Computational Complexity of the MDS Algorithm}
If the network is extremely dense (e.g., fully connected with weight 1.0) and the density threshold is extremely low (e.g. 0.0), the depth-first search will exhaust all the possible subgraphs. Therefore, the worst-case theoretical time complexity is exponential. However, real PPI networks are quite sparse. As a result, in practice, our MDS algorithm is very efficient. ‬‬‬Moreover, as follows, we also provide the MDS-Scale (i.e., partition-enabled MDS), which is based on graph partitions and thus allowing fully parallelization.

\section{Scaling Up the Algorithm Based on Partitioning}
For a large-scale graph, we will first partition it into several pieces. On each piece, the mining algorithm is executed to discover the diversified maximal dense subgraphs. After that, the results from each partition are merged into one set to be diversified.
For the partitioning, we use the measure cohesiveness \cite{nepusz2012detecting} to roughly cut the graph into several parts, with each of which being no bigger than a size limit. When merging the results from each partition, we first extend each maximal dense subgraph based on density, to avoid the effect of partitioning. If one subgraph has been covered by $80\%$ of another extended subgraph, we will not retain it in the result set. After that, we merge the extended results from each partition into one set and diversify it by $cov/k$. The algorithm is as follows.

The main procedure MDS-Scale is in appendix.\\

\section{Result}
In this section, we evaluate the effectiveness and efficiency of our method, compare with the state-of-art methods on a range of biological networks from different species. We use several real world networks from two different organisms, one category is five PPI datasets for yeast, and the other is the genetic network for human. The statistics of these datasets is in Table \ref{tab:stats} of Appendix. 

\subsection{Experiment setup}

\subsubsection{Gold standard}
For PPI networks, we compare the detected complexes to four reference complex sets: the first derived from the MIPS catalog of protein complexes and the second from Gene Ontology-based complex annotations in the SGD. 
The Compleat gold standard is from http://www.flyrnai.org/compleat/. 
The Human-String ground truth (String-GT) is gained from: \\
http://mips.helmholtz-muenchen.de/genre/proj/corum/. Table~\ref{tb:gold} illustrates some properties of the gold standards.


\begin{table}
\centering
\caption{Statistics of the gold standards}
\label{tb:gold}
\begin{tabular}{|c|c|c|c|c|}
\hline 
 & MIPS & SGD & COMPLEAT & String-GT\tabularnewline
\hline 
\hline 
Number of nodes  & 1189 & 1279 & 2877 & 1700 \tabularnewline
\hline 
Number of complexes  & 203 & 323 & 1317 & 840 \tabularnewline
\hline 
Overlapping complex pairs  & 401(2\%) & 296(0.6\%) &9845 (1.1\%) & 7443(2\%) \tabularnewline
\hline 
\end{tabular}

\end{table}

\subsubsection{Evaluation metrics} 
We assess the quality of the predicted complexes by three scores defined in [16]: fraction (frac), accuracy (acc), and the maximum matching ratio (mmr). frac is defined to be the fraction of pairs between predicted and reference complexes with an overlap score no less than 0.25; acc is the geometric mean of two other measures, namely the clustering-wise sensitivity (Sn) and the clustering-wise positive predictive value (PPV). See \cite{nepusz2012detecting} for a precise description of them. mmr is based on a maximal one-to-one mapping between predicted and reference complexes. With a bipartite graph, in which the two sets of nodes represent the reference and predicted complexes respectively, and an edge is weighted by the overlap score between two complexes, mmr is the total weight of the selected edges which represent an optimal one-to-one matching between the two sets, divided by the number of reference complexes. 

For all these metrics, we compare under different density thresholds with COACH, MCL and ClusterONE, as they can use density as a parameter. Besides these metrics, we also do comparisons on coverage. 

\subsubsection{Baseline for comparison} We choose the most prominent methods for detecting protein complexes, MCODE\cite{bader2003automated}, RNSC\cite{king2004protein}, MCL\cite{van2001graph}\cite{van2008graph}, AP\cite{frey2007clustering}, CMC\cite{liu2009complex}, COACH\cite{wu2009core} and ClusterONE\cite{nepusz2012detecting} as the baseline. For easy to understand, we first briefly review ClusterONE. For the other baselines, please refer to ``Other baseline methods" section in appendix for detailed introduction. 

ClusterONE consists of three major steps. First, starting from a single seed vertex, a greedy procedure adds or removes vertices to find groups with high cohesiveness. The growth process is repeated from different seeds to form multiple, possibly overlapping groups. Second, merge those pairs of groups for which the overlap score is above a specified threshold. Third, discard complex candidates that contain less than three proteins or whose density is below a given threshold. We can see, the density metric is used as a post-processing filter of the method, which is not as direct as our density metric.

\subsection{Evaluations}
\subsubsection{Effectiveness}

\begin{figure}
  \centering
  \subfloat[]{\includegraphics[width=0.46\textwidth]{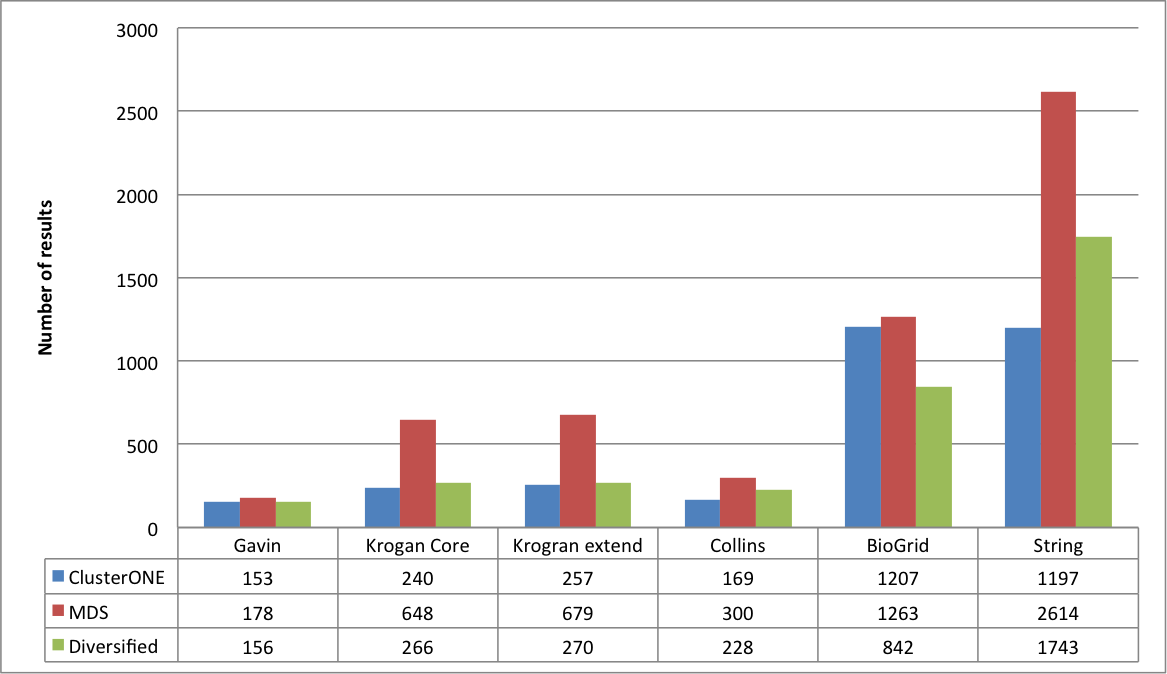}\label{fig:number}}
  \hspace{0.5cm}
  \subfloat[]{\includegraphics[width=0.5\textwidth]{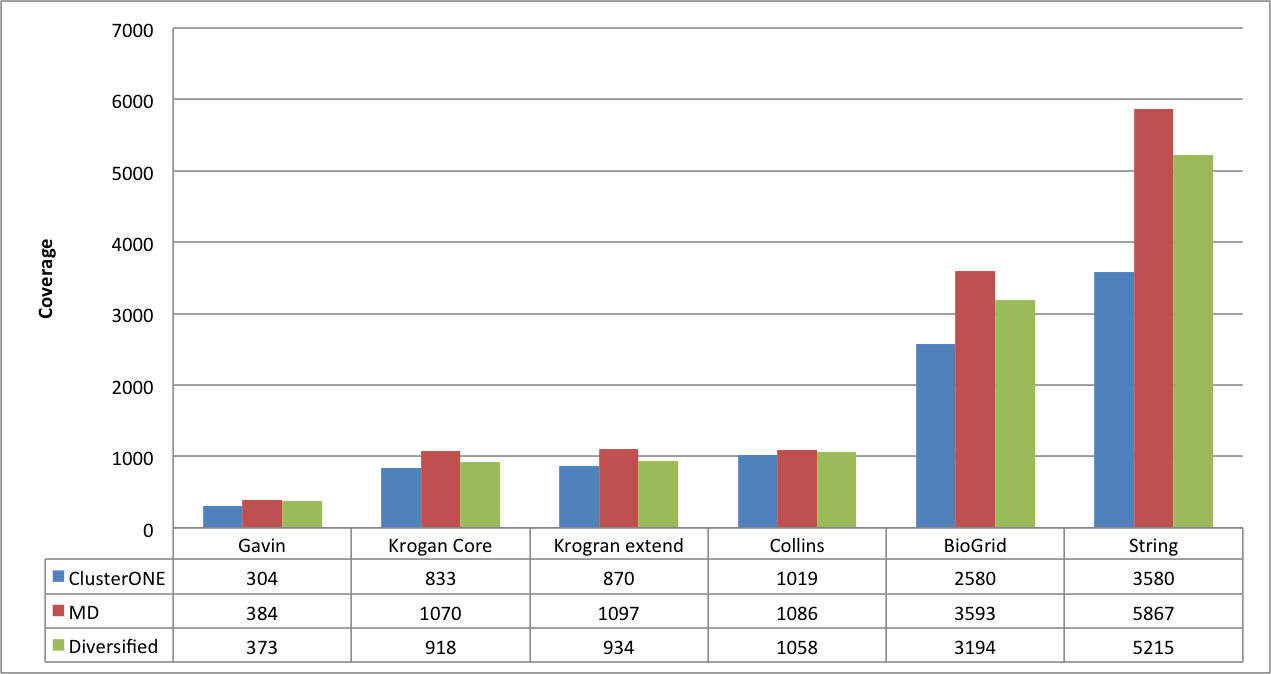}\label{fig:coverage}}
  
  \caption{Number of results and coverage for different methods on all the datasets under different density thresholds}\label{fig:all_coverage}
 \end{figure}

In this section, we present the number of results for different methods on different datasets under different density thresholds. The evaluation is under density thresholds between 0.5 and 0.9. We visualize density 0.6 for number of results in Fig. \ref{fig:number} and density 0.5 for coverage in Fig. \ref{fig:coverage}. The detailed results for other thresholds are quite similar and shown in Table \ref{tab:numresult} and Table \ref{tab:covaragenew} in Appendix. As shown, our method finds more complexes than all the other methods on all the datasets and density threshold. This is because that, comparing with the local search heuristics in ClusterONE and other clustering-based methods, our method is based on enumeration with smart ordering based on candidates, which can discover more results. However, pruning can still guarantee the efficiency, and diversifying the results can remove the redundancy as well. For all the datasets except Collins, `MDS' represents using our method `Maximal Dense Subgraphs Mining' with leftmost branch only, and `Diversified' represents `$cov/k>=4$'. For Collins, `Diversified' represents `$cov/k>=5$'.

\subsubsection{Performance}
We present the three quality scores obtained using both the MIPS, SGD and Compleat reference set in Figure~\ref{fig:mips_sgd} and \ref{fig:merge}. Shades of the same color denote individual quality scores; the total height of each bar is the value of the composite score. Larger scores are better, and the sum of the three scores is a composite score. 

We  compare the scores on unweighted network, which is the binary version of the weighted datasets, illustrated in \ref{fig:f9} and \ref{fig:f10}. We are capable of matching more complexes with a higher accuracy and providing a better one-to-one mapping with reference complexes in almost all the 4 PPI datasets (except collins) comparing to other baseline methods. We don't show the result of Biogrid since many baselines can not finish it efficiently. 
For weighted PPI networks, please refer to Figure \ref{fig:mips_sgd} in appendix.



\begin{figure}[!tbp]
  \centering
  \subfloat[]{\includegraphics[width=0.45\textwidth]{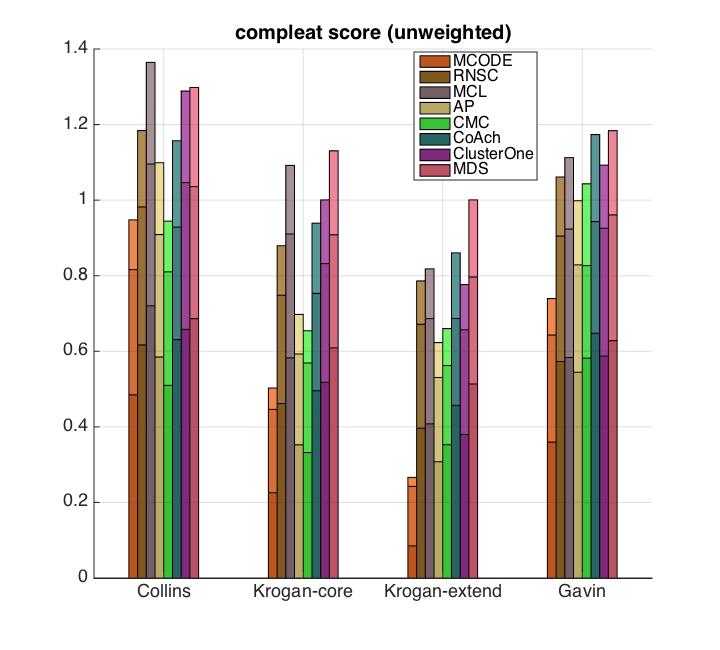}\label{fig:f9}}
    \hspace{1cm}
  \subfloat[]{\includegraphics[width=0.45\textwidth]{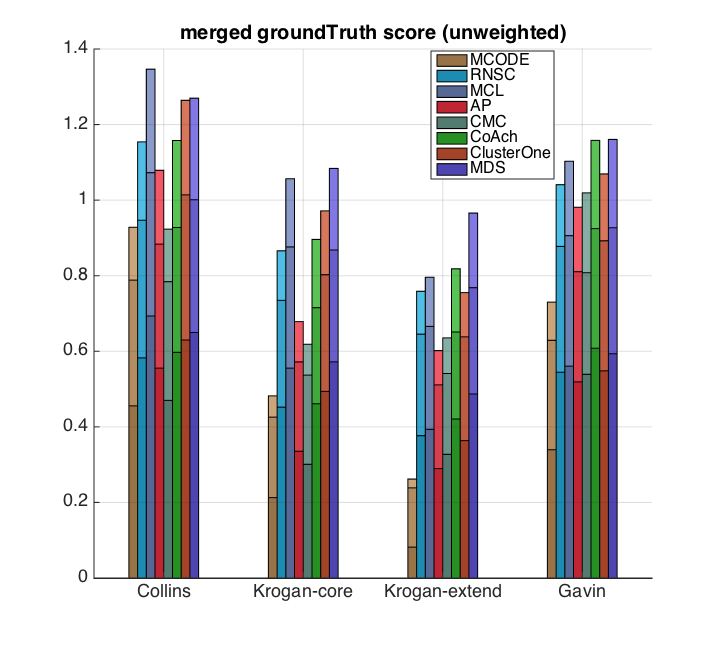}\label{fig:f10}}
  \caption{Results (bottom-up: frac, acc, mmr) of various methods on 4 PPI unweighted datasets using compleat and a mixture of 3 gold standards.}\label{fig:merge}
\end{figure}


For the PPV score, as mentioned in supplemental material of \cite{nepusz2012detecting}, a perfect clustering algorithm that always returns the reference complexes from the data may have a lower positive predictive value than a dummy algorithm which places every protein in a separate cluster. 

The accuracy measure explicitly penalizes predicted complexes that do not match any of the reference complexes. However, gold standard sets of protein complexes are often incomplete \cite{nepusz2012detecting}. As a consequence, predicted complex not matching any known reference complexes may still exhibit high functional similarity or be highly co-localized, and therefore they could still be prospective candidates for further in-depth analysis. In other words, a predicted complex that does not match a reference complex is not necessarily an undesired result.We evaluate the functional homogeneity of some detected complexes by literature search as follows.

We find the following complex in Gavin under density threshold 0.3:
(YBR265W YPR035W YDR028C YBR126C YNL076W YMR261C YBL023C YML100W YDR074W) which cannot be found in any gold standard.

We search the literature and find this complex included in the excel file ``1752-0509-3-74-S3.XLS'' from this \href{https://www.repository.cam.ac.uk/}{link}:
``Set of 491 complexes originating from Gavin et al., Nature 2006''.
It is included in Trehalose-6-phosphate synthase/ phosphatase complex:
YBL023C YBR265W YNL076W YPR035W YBR126C YDR028C YDR074W YDR171W YJL138C YJR138W YKR059W YLR075W YML100W YMR261Cs

It also appears in the ``Process GO Term physiological processes and related genes'' from this \href{http://gonet.genomics.ics.uci.edu/pgo/p1_5_115.html}{link}. It gives out the set of interacting Genes of YBR126C, which includes all the proteins in the above complex we find.

\subsubsection{Performance on Human genetic network}
We present the results on String dataset. By partitioning algorithm, we first partition the dataset, and discover the maximal dense subgraphs in each partition, and then merge the results into one whole, and then do diversifying (cov/k). We present our performance score of String dataset with different baseline methods. 

$I$ is the parameter inflation for MCL algorithm, which tunes the granularity of the clustering. Larger inflation values result in smaller clusters, while smaller inflation values generate only a few large clusters.


\begin{table}
\caption{Comparison of other methods and our methods on String datasets.}
\centering
\begin{tabular}{|c|c|c|c|c|}
\hline 
 & MCL: $I = 3.5$  & ClusterOne: $\theta = 0.6$  & MDS $\theta = 0.8$  & Diverse $cov/k >= 4, \theta = 0.8$\tabularnewline
\hline 
\hline 
frac & 0.0216 & 0.2938 & 0.5336 & 0.476\tabularnewline
\hline 
acc & 0.2311 & 0.3704 & 0.3341 & 0.3233\tabularnewline
\hline 
mmr & 0.0086 & 0.0884 & 0.1693 & 0.1332\tabularnewline
\hline 
sum & 0.2613 & 0.7526 & 1.037 & 0.9325\tabularnewline
\hline 
\end{tabular}
\end{table}

\subsubsection {Scalability} We evaluate the efficiency and scalability of our algorithm on various datasets. Fig. \ref{fig:t1} illustrates the linearity of our MDS-Scale method v.s. MDS along with  PPI networks of different sizes. Fig. \ref{fig:t2} shows that the running time of MDS-Scale method is more stable than MDS on different density threshold.

\begin{figure}
  \centering
  \subfloat[]{\includegraphics[width=0.35\textwidth]{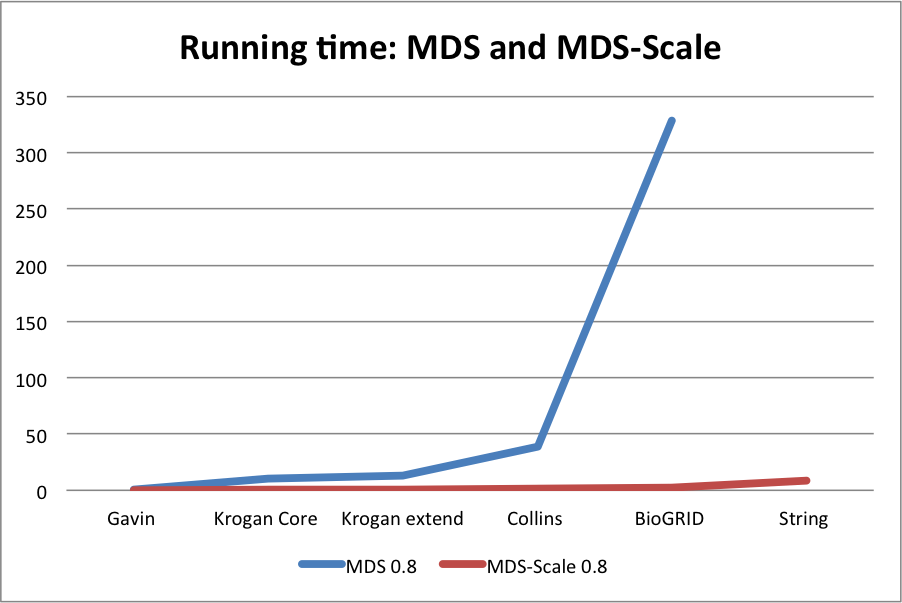}\label{fig:t1}}
    \hspace{1cm}
  \subfloat[]{\includegraphics[width=0.35\textwidth]{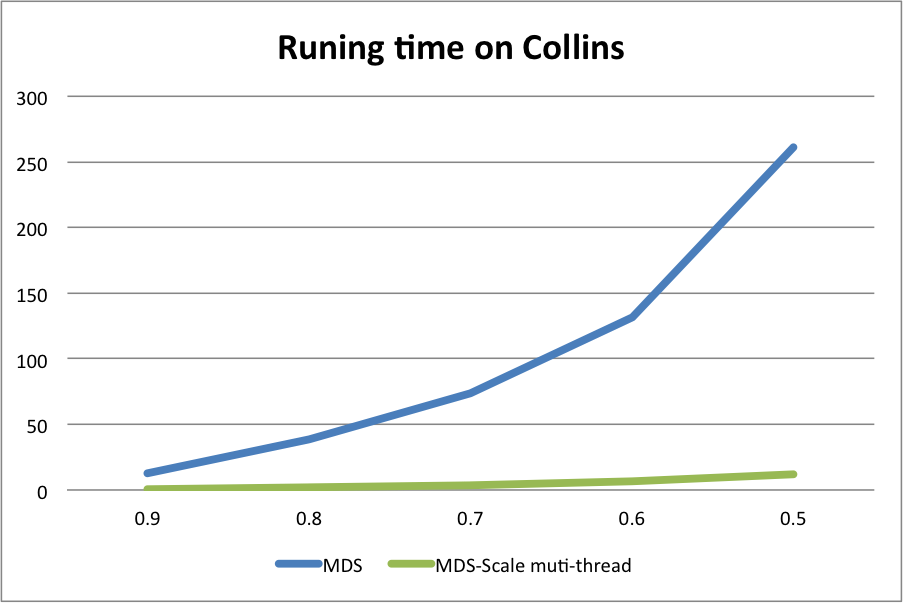}\label{fig:t2}}
  \caption{Comparison of the runtime (seconds) of our methods on various biological networks (a) and under different density threshold (b).}

\end{figure}
The partitions numbers in our evaluation are illustrated as follows in Table. \ref{tab:partnum}. The size and number of partitions can be easily adjusted for different computing systems. 

\begin{table}
\caption{Partitions numbers for different datasets.}
\label{tab:partnum}
\centering
\begin{tabular}{ | l | l | l | l | l | l | l | }
\hline
& Gavin & Krogan Core & Krogan Extend & Collins & BioGRID & String\\
\hline
No. Partitions & 220 & 639 & 1241 & 254 & 1483 & 4725 \\
\hline
  
\end{tabular}

\end{table}

\section{Discussion}
We have modeled the problem of detecting complexes in biological networks as discovering the diversified maximal dense subgraphs. With the edge density measure, the dense subgraphs, that are complexes, can be defined explicitly and flexibly. Based on this, we seek all the dense subgraphs of maximal cardinality with at least the specified density. Meanwhile, the results without enough diversity are sieved along the enumeration. We also scale up the algorithm by partitioning method. We evaluate the effectiveness and efficiency of our method on a diverse set of interaction networks from different species, five PPI networks for yeast and one genetic network for human. Extensive experiments show our results have better correspondence with reference complexes in MIPS, SGD, Compleat and the ground-truth database for String than the state-of-arts.

In the future, we plan to pursue further improvements in larger-scale and denser networks. We believe the algorithm is general in various domains, such as identifying groups of densely interconnected nodes in social networks or word association networks. We also hope to explore other biological network-based applications, such as identifying homology relationships between sequences and orthology inference across multiple species.

\subsubsection*{Supplementary Data and Availability.}
The source code and supplementary data are available at
https://github.com/zgy921028/MDSMine\\

\bibliographystyle{plain}

\bibliography{ref}
\newpage
\section*{Appendix}
\vspace{0cm}

\vspace{1cm}
\subsection*{Proof of Theorem 3}
\begin{proof}

The conclusion holds if and only if \\
$deg(i)|_{S\cup\{k\}}$ is no greater than the average degree of $S\cup \{k\} \cup \{i\}$. 
The average degree of $S\cup \{k\} \cup \{i\}$ is 
$$\frac{2 (|E(S)| + deg(k)|_S + deg(i)|_S + w(i,j))}{|S| + 2}$$

So we need to prove $deg(i)|_{S\cup\{k\}} \leq \frac{2 (|E(S)| + deg(k)|_S + deg(i)|_S + w(i,j))}{|S| + 2}$\\

(1) When $i$ and $k$ have no edge.\\

As $deg(i)|_S \leq \frac{2 (|E(S)| + deg(i)|_S )}{|S| + 1}$, we have 
$2 |E(S)| \geq (|S| -1) deg(i)|_S$.\\

Similarly, we have  $2 |E(S)| \geq(|S| -1) deg(k)|_S$.\\

When $deg(i)|_S \leq 2 deg(k)|_S$ or $|S| deg(i)|_S \leq (|S|+1) deg(k)|_S$ is satisfied, the theorem holds.\\

(2) When there is an edge between i and k.
If $deg(i)|_S + |S| \leq 2 deg(k)|_S$ is satisfied, the theorem is proved. \\

Or, if $|S| deg(i)|_S + |S| \leq (|S| + 1)deg(k)|_S$ is satisfied, the theorem is also proved. \\

In most cases, the conditions are satisfied, thus the theorem is proved. 
\end{proof}

\newpage
\subsection*{Supplementary Table I}
\vspace{2cm}
\begin{table}
\centering

\begin{tabular}{|c|c|c|c|c|c|c|}
\hline 
 & Gavin & Krogan Core & Krogan extended & Collins & BioGrid & Human-String\tabularnewline
\hline 
\hline 
$|V|$ & 1855  & 2708 & 3672  & 1622  & 5640  & 16384\tabularnewline
\hline 
$|E|$ & 7669  & 7123  & 14317  & 9074 & 59748 & 419383\tabularnewline
\hline 
Weighted & Y & Y & Y & Y & N & Y\tabularnewline
\hline 
\end{tabular}
\caption{Statistics of yeast PPI datasets: the number of vertices and edges.}
\label{tab:stats}
\end{table}

\vspace{2cm}
\begin{table}
\centering

\begin{tabular}{ | l | l | l | l | l | l | l | }
\hline
  Dataset & Method & 0.9 & 0.8 & 0.7 & 0.6 & 0.5 \\ \hline \hline
  Gavin & MDS & 2 & 20 & 50 & 105 & 178 \\ \hline
   & Diverse & 1 & 5 & 28 & 94 & 156 \\ \hline
   & ClusterONE & 0 & 5 & 34 & 89 & 153 \\ \hline
  Krogan Core        & MDS & 211 & 232 & 260 & 382 & 648 \\ \hline
   & Diverse & 121 & 157 & 185 & 230 & 266 \\ \hline
   & ClusterONE & 79 & 111 & 139 & 183 & 240 \\ \hline
  Krogan extended  & MDS & 204 & 222 & 272 & 397 & 679 \\ \hline
   & Diverse & 123 & 157 & 189 & 234 & 270 \\ \hline
   & ClusterONE & 82 & 112 & 142 & 200 & 257 \\ \hline
  Collins & MDS & 189 & 225 & 235 & 240 & 300 \\ \hline
  (cov/k:5) & Diverse & 141 & 170 & 193 & 212 & 228 \\ \hline
   & ClusterONE & 100 & 117 & 128 & 136 & 169 \\ \hline
  BioGRID & MDS & 897 & 984 & 1061 & 1198 & 1263 \\ \hline
   & Diverse & 584 & 656 & 707 & 799 & 842 \\ \hline
   & ClusterONE & 178 & 241 & 263 & 473 & 1207 \\ \hline
  String & MDS & 774 & 1141 & 1508 & 1956 & 2614 \\ \hline
   & Diverse & 516 & 761 & 1006 & 1304 & 1743 \\ \hline
   & ClusterONE & 165 & 294 & 422 & 721 & 1197 \\ \hline
  
\end{tabular}
\caption{Number of results for different methods on all the datasets under various density thresholds}
\label{tab:numresult}

\end{table}

\newpage
\subsection*{Supplementary Table II}
\vspace{5cm}
\begin{table}
\centering

\begin{tabular}{ | l | l | l | l | l | l | l | }
\hline
  Dataset & Method & 0.9 & 0.8 & 0.7 & 0.6 & 0.5 \\ \hline
  Gavin & MDS & 6 & 60 & 169 & 384 & 644 \\ \hline
   & Diverse & 3 & 19 & 111 & 373 & 622 \\ \hline
   & ClusterONE & 0 & 17 & 107 & 304 & 585 \\ \hline
  Krogan Core     & MDS & 573 & 701 & 814 & 1070 & 1443 \\ \hline
   & Diverse & 483 & 626 & 739 & 918 & 1061 \\ \hline
   & ClusterONE & 349 & 496 & 652 & 833 & 1048 \\ \hline
  Krogan extended  & MDS & 572 & 691 & 837 & 1097 & 1487 \\ \hline
   & Diverse & 491 & 626 & 754 & 934 & 1078 \\ \hline
   & ClusterONE & 354 & 495 & 658 & 870 & 1064 \\ \hline
  Collins & MDS & 750 & 901 & 1005 & 1086 & 1210 \\ \hline
  (cov/k:5) & Diverse & 702 & 846 & 963 & 1058 & 1138 \\ \hline
   & ClusterONE & 629 & 809 & 915 & 1019 & 1147 \\ \hline
  BioGRID & MDS & 2689 & 2951 & 3181 & 3593 & 3788 \\ \hline
   & Diverse & 2335 & 2623 & 2827 & 3194 & 3367 \\ \hline
   & ClusterONE & 1151 & 1580 & 1776 & 2580 & 4418 \\ \hline
  String & MDS & 2321 & 3422 & 4523 & 5867 & 7841 \\ \hline
   & Diverse & 2063 & 3042 & 4021 & 5215 & 6970 \\ \hline
   & ClusterONE & 970 & 1642 & 2381 & 3580 & 5618 \\ \hline
\end{tabular}
\caption{Coverage (Number of vertices covered) of different methods on all the datasets under various density thresholds}
\label{tab:covaragenew}

\end{table}

\newpage
\subsection*{Partition algorithm}
\vspace{5cm}
\IncMargin{1em}
\begin{algorithm}[H]
\SetKwData{Left}{left}\SetKwData{This}{this}\SetKwData{Up}{up}
\SetKwFunction{Union}{Union}\SetKwFunction{partition}{partition}
\SetKwFunction{Union}{Union}\SetKwFunction{extend}{extend}
\SetKwFunction{Union}{Union}\SetKwFunction{covK}{covK}

\SetKwInOut{Input}{input}\SetKwInOut{Output}{output}
\Input{$G(V,E), theta$ (density threshold), $m, k, \beta$ (cov/k threshold)}
\Output{List of maximal dense subgraphs}
\BlankLine

$rPlst, sPlst=$  \partition $(G(V,E), threSize, maxPartSize)$
\BlankLine

\ForEach {$P\in rPlst$}
{
  $SortByDeg{V(P)}$
  
  \ForEach {$vi \in V(P)$}
  {
    $mdsList = MDS(P(V,E), \theta, m, k)$
  }
}
\BlankLine

\ForEach {$P'\in sPlst$}
{
  \If {$density(Q') > \theta$}
  {
    $mdsList.add(P')$
  }
}

\extend(mdsList)
\BlankLine
\Return \covK(mdsList, $\beta$)

\caption{{\sc MDS-Scale Main}}
\end{algorithm}\DecMargin{1em}

\newpage
\subsection*{Other baseline methods}
\vspace{1cm}
\IncMargin{1em}
The MCODE algorithm \cite{bader2003automated} consists of three phases: vertex weighting, protein complex formation and post-processing. The vertex weighting phase assigns a score to each vertex measuring the “cliquishness” of the neighborhood of the vertex. Protein complexes are then grown from each vertex, starting from the one with the highest weight. Finally, there are two possible post-processing operations: haircut, which iteratively removes vertices that are connected by only a single edge to the rest of the complex, and fluffing, which tries to expand the complex with other vertices if they connect to many vertices of the same complex. MCODE is able to produce overlapping complexes in the fluffing phase.

The RNSC algorithm \cite{king2004protein} is Restricted Neighborhood Search Clustering algorithm, which partitions the network’s node set into clusters based on a cost function that is assigned to each partitioning. It then filtered the RNSC output so that only clusters that share characteristics of known protein complexes are considered. The RNSC searches for a low-cost clustering by first composing an initial random clustering, then iteratively moving one node from one cluster to another in a randomized fashion to improve the clustering’s cost. RNSC is essentially a cost-based local search algorithm.

MCL \cite{van2008graph} is a graph clustering algorithm, whose engine is the Markov Cluster Process or the MCL process. The MCL process takes a stochastic matrix as input, and then alternates expansion and inflation, each step defining a stochastic matrix in terms of the previous one. The process is typically applied to the matrix of random walks on a given graph G, and the connected components of the graph associated with the process limit generically allow a clustering interpretation of G.

AP \cite{frey2007clustering} simultaneously considers all data points as potential exemplars, then recursively transmits real-valued messages along edges of the network until a good set of exemplars and corresponding clusters emerges. Messages are updated on the basis of simple formulas that search for minima of an appropriately chosen energy function. At any point in time, the magnitude of each message reflects the current affinity that one data point has for choosing another data point as its exemplar. So, the method is called “affinity propagation”.

CMC \cite{liu2009complex} is a Clustering method based on Maximal Cliques (CMC) to detect protein complexes. CMC first obtains all the maximal cliques, then assigns each interaction a score based on a reliability measure. Therefore, each clique can be scored with its weighted density. Last, CMC removes or merges highly overlapping cliques to generate protein complexes. In particular, if two cliques are highly overlapping, CMC either merges these two cliques as a bigger one or simply removes the one with a lower score (weighted density) depending on their inter-connectivity.

COACH \cite{wu2009core}. To provide insights into the organization of protein complexes, Wu et al.\cite{wu2009core}  presents a COre-AttaCHment based method (COACH) which detects protein complexes in two stages. In the first stage, COACH defines core vertices from the neighborhood graphs and then detects protein-complex cores as the hearts of protein complexes. In the second stage, COACH includes attachments into these cores to form biologically meaningful structures.

\newpage
\subsection*{Supplementary Figure}
 \begin{center}
\vspace{5cm}

Our MDS and diversified method outperformed all the other density-based baselines by choosing 0.7 as our density threshold on weighted PPI networks. 
\begin{figure}
  \subfloat[]{\includegraphics[width=0.45\textwidth]{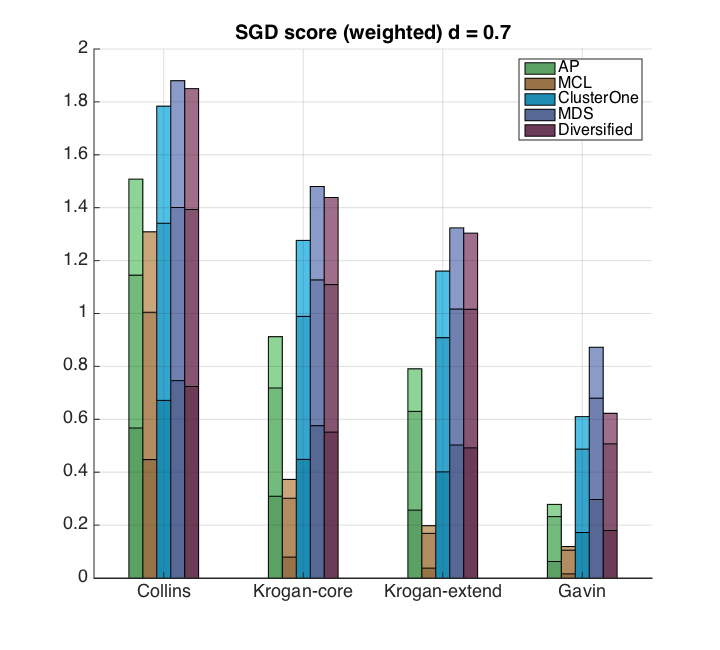}\label{fig:sgd07}}
  \hspace{1cm}
  \subfloat[]{\includegraphics[width=0.45\textwidth]{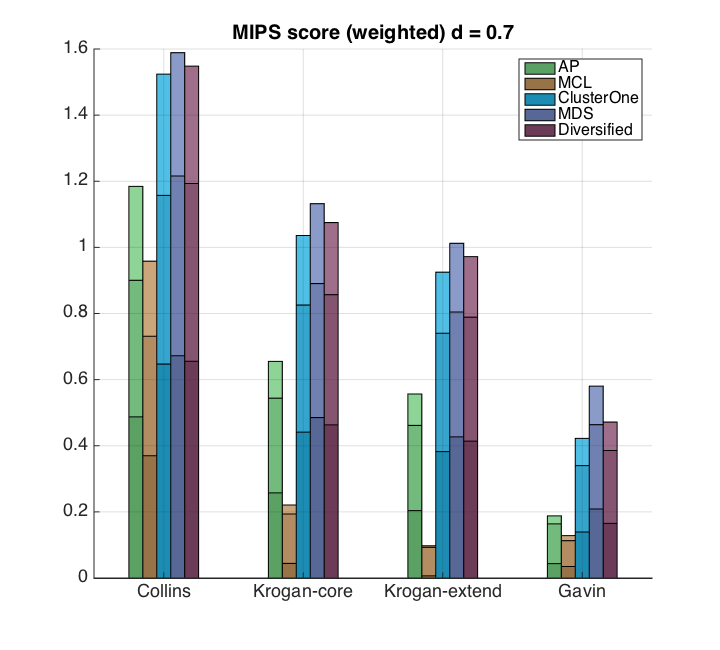}\label{fig:mips07}}
  \caption{Results (bottom-up: frac, acc, mmr) of various methods on 4 PPI weighted datasets using MIPS (a) and SGD(b) gold standard.}  \label{fig:mips_sgd}
\end{figure}
\end{center}
\end{document}